\documentclass[11pt]{article}

\usepackage{amsmath,amsthm,verbatim,amssymb,amsfonts,amscd, graphicx}
\usepackage[utf8]{inputenc}
\usepackage{graphics}
\usepackage{hyperref}
\hypersetup{
    colorlinks=true,   	
    linkcolor=red,      
    citecolor = [rgb]{0 0.7 0},   	
    filecolor=magenta, 	
    urlcolor=blue
}
\usepackage{enumitem}
\usepackage{apptools}
\usepackage{titlesec}
\usepackage[baseline]{euflag}
\usepackage{authblk}
\usepackage{appendix}
\usepackage{xcolor}
\AtAppendix{\counterwithin{lemma}{section}}

\theoremstyle{plain}

\newtheorem{theorem}{Theorem}
\newtheorem{corollary}{Corollary}
\newtheorem{lemma}{Lemma}
\newtheorem{proposition}{Proposition}

\theoremstyle{definition}

\newcommand{\Var}{\mathrm{Var}}
\newcommand{\RL}{\mathbb{R}}
\newcommand{\snr}{{\sf snr}}

\newcommand{\p}[1]{
{#1}^{\prime}
}

\newcommand{\Hrm}{{\rm H}}

\newcommand{\R}{\mathbb{R}}
\newcommand{\E}{\mathbb{E}}

\newcommand{\ra}{\rightarrow}

\begin{document}
 
\title{The entropic doubling constant\\
and robustness of Gaussian codebooks for additive-noise channels}
\author[1]{Lampros Gavalakis%
\thanks{L.G. has received funding from the 
  European Union's Horizon 2020 research and innovation program
  under the Marie Sklodowska-Curie grant agreement No 101034255 \euflag\ 
  and by the B{\'e}zout Labex, funded by ANR, reference ANR-10-LABX-58.}}
\author[2]{Ioannis Kontoyiannis}
\author[3]{Mokshay Madiman}
\affil[1]{\small Univ Gustave Eiffel, Univ Paris Est Creteil, CNRS, LAMA 
UMR8050 F-77447 Marne-la-Vall{\'e}e, France}
\affil[2]{Statistical Laboratory, Univ of Cambridge,
Wilberforce Road, Cambridge CB3 0WB, U.K.}
\affil[3]{Department of Mathematical Sciences, Univ of Delaware, 
501 Ewing Hall, Newark, DE 19716, U.S.A.}

\maketitle

\begin{abstract}
Entropy comparison inequalities are obtained for the differential 
entropy $h(X+Y)$ of the sum of two independent random vectors
$X,Y$, when one is replaced by a Gaussian.
For identically distributed random vectors $X,Y$, 
these are closely related to bounds on 
the entropic doubling constant, which quantifies the 
entropy increase when adding an independent copy of a random 
vector to itself. Consequences of both large and small doubling are
explored. For the former, lower bounds are deduced on the entropy 
increase when adding an independent Gaussian, while for the latter, 
a qualitative stability result for the entropy power 
inequality is obtained.
In the more general case of
non-identically distributed random vectors
$X,Y$, a Gaussian comparison inequality 
with interesting implications for channel coding
is established:
For additive-noise channels with a power constraint, 
Gaussian codebooks come within a
$\frac{\snr}{3\snr+2}$ factor of capacity.
In the low-SNR regime
this improves the half-a-bit additive bound 
of Zamir and Erez (2004). 
Analogous results
are obtained for additive-noise multiple
access channels, and for linear, additive-noise
MIMO channels.
\end{abstract}


\newpage

\section{Introduction}
The (differential) {\it entropy} 
of an $\RL^{d}$-valued random vector $X$ with density $f$
is defined as,
$$
h(X)=-\int_{\RL^d} f(x) \log f(x) \,dx ,
$$
whenever the integral exists. When $X$ is supported on a 
strictly lower-dimensional subset of $\RL^d$ or, more generally,
when the law of $X$
does not have a density with respect to the
$d$-dimensional Lebesgue measure, 
we set $h(X) =-\infty$. The {\it entropy power} of $X$ is 
defined by
$N(X)=e^{\frac{2}{d}h(X)}$ and takes values in $[0,\infty]$. 

\subsection{Small and large entropic doubling}
\label{s:smalllarge}
The (entropic) 
\textit{doubling constant} of a random vector $X$
in $\R^d$
is defined as,
$$
\sigma[X] := \frac{N(X+\p{X})}{2N(X)}
=\frac{1}{2}\exp\Big\{\frac{2}{d}[h(X+X')-h(X)]\Big\},
$$
where $X, \p{X}$ are independent and identically distributed 
(i.i.d.). 

Observe that $\sigma[X]$ is a simple monotonic function
of the entropy increase $h(X+X')-h(X)$.
By the entropy power inequality
(EPI), $\sigma[X] \geq 1$ for any random vector $X$,
with equality if and only if $X$ is Gaussian. 
The study of the doubling constant is obviously intimately related
to the study of the EPI, in particular
in the context of the central limit
theorem~\cite{barron:clt,johnson-barron:04},
as well as in connection with one of the discrete
EPI discrete analogues
put forward by Tao~\cite{tao:10}. 
For continuous random variables, the doubling
constant and its connection with entropy inequalities
motivated by additive combinatorics are discussed
extensively in~\cite{KM:14,KM:16}.
Applications of bounds on the doubling constant
are described in the following two subsections.

From the definition, we see that
$h(X+\p{X}) -h(X)=\frac{d}{2} \log (2\sigma[X])$. 
If $Z$ is an independent Gaussian with the same covariance matrix as $X$, 
the maximum entropy property of the Gaussian trivially implies that,
$$
h(Z)\geq h\Big(\frac{X+\p{X}}{\sqrt{2}}\Big) = h(X+\p{X}) -\frac{d}{2}\log 2.
$$
Therefore, $h(Z)-h(X)\geq \frac{d}{2} \log\sigma[X]$, and since
$h(X+Z)\geq h(Z)$, we have:
\begin{equation}
h(X+Z)-h(X) \geq \frac{d}{2} \log\sigma[X].
\label{eq:intro}
\end{equation}
This says that adding a Gaussian with the same covariance 
matrix increases the entropy of $X$ by an amount that can
be bounded below in terms of the doubling constant. 
Our first result gives a quantitatively sharper version
of that statement.

\begin{theorem}[Large doubling]
\label{thm:quant}
Suppose $X$ is a  random vector in $\R^d$,
and $Z$ is an independent Gaussian random vector with 
the same covariance matrix as $X$. Then,
\begin{equation*}
h(X+Z)-h(X) \geq  \frac{d}{2}\log \big[ 
\max \big\{ \varphi\, \sigma[X] , 1 + \sigma[X] \big\} \big] ,
\label{combined-doublingbound}
\end{equation*}
where $\varphi=\frac{1+\sqrt{5}}{2}>1$ is the golden ratio.
\end{theorem}

We will prove the two parts of the theorem separately,
in Section~\ref{s:main}.
The first bound is simply a rewriting
of Theorem~\ref{thm:ENTqn},
and the second bound is a special case 
of Theorem~\ref{prop:doub}, corresponding
to $K_Z=K_X$ there.
Moreover, the result of Theorem~\ref{prop:doub},
and by extension Theorem~\ref{thm:quant},
is sharp in the sense that equality holds if $X$ is Gaussian.

Theorem~\ref{thm:quant} is about the consequences 
large doubling.
In the reverse direction, next, we explore the consequences of small doubling. 
First, we note that stability of the lower bound 
$\sigma[X]\geq 1$ implied by the EPI may fail for
``strong'' distance measures like the quadratic
Wasserstein distance
or relative entropy. That is, it may be possible
to have $\sigma[X]$ arbitrarily close to 1,
while the distance of its distribution from 
the Gaussian remains bounded away from zero;
see the discussion in~\cite{courtade:18}.

But, as we show in Theorem~\ref{smalldoublinglevy}
stated next, 
if we impose a mild moment assumption,
a more qualitative version
of stability in the sense of weak convergence
does in fact hold. Also, under much stronger 
assumptions on the densities considered,
a different stability result for the
quadratic Wasserstein distance is established
in~\cite{courtade:18}.

Our stability result in Theorem~\ref{smalldoublinglevy} 
below follows from Theorem~\ref{smalldoublinglevy2}
in Section~\ref{s:main},
which is established
by combining an earlier
result of Carlen and Soffer~\cite{carlen:91}  
with the stability of Cram\'{e}r's theorem 
due to L{\'e}vy~\cite{levy:book}.

For a random vector $X$ with law $\nu$, denote,
$$
\psi_{X}(R) := \mathbb{E}\big(\|X\|^2; \|X\| > R\big)
=\int_{\|x\|>R} \|x\|^2 d\nu(x),
\qquad R \geq 0,
$$
where $\|\cdot\|$ denotes the usual Euclidean norm,
$\|x\|=[x_1^2+x_2^2+\cdots+x_d^2]^{1/2}$,
$x\in\RL^d$.
We may think of $\psi_{X}$ as the ``variance profile'' of $X$.
If $X$ has density $f$, we also write $\psi_f$ for $\psi_X$.
By elementary reasoning, the function 
$\psi_X:[0,\infty)\ra [0,\infty]$ is non-increasing. Moreover, if $X$ 
has a density $f$ and finite variance, then 
$\psi_X(R)=\psi_f(R)=\int_{\|x\|>R} \|x\|^2 f(x) dx$ 
is continuous in $R$, with $\psi_X(0)=\E(\|X\|^2)$ 
and $\lim_{R\ra \infty} \psi_X(R)=0$.
Let $\mathcal{F}_d$ be the class of probability density 
functions on $\RL^d$, namely, of all functions
$f:\RL^d\ra[0,\infty)$ such that $\int_{\RL^d} f(x) dx=1$.
For a given function $\psi:[0,\infty)\ra [0,\infty)$ with 
$\lim_{R\ra \infty} \psi(R)=0$, we define the class of 
densities ${\cal C}_\psi\subset{\cal F}_d$ as:
\begin{equation}
\mathcal{C}_{\psi}:=\big\{ f\in \mathcal{F}_d: \psi_f \leq \psi\big\},
\label{eq:Cpsi}
\end{equation}
where the inequality $\psi_f\leq \psi$ is understood to be pointwise.
Write $\gamma_{\mu,\Sigma}$ for the Gaussian measure on 
$\mathbb{R}^d$ with mean $\mu$ and covariance $\Sigma$, and 
$d_{\rm LP}$ 
for the L\'evy--Prokhorov metric on the space of probability 
measures (see Section \ref{sec:prelim} for the precise definition).

\begin{theorem}[Small doubling]
\label{smalldoublinglevy}
Suppose $X,\p{X}$ are i.i.d. random vectors with law $\nu$ 
on $\mathbb{R}^d$, which has density in $\mathcal{C}_{\psi}$
for some $\psi:[0,\infty)\ra [0,\infty)$ with 
$\lim_{R\ra \infty} \psi(R)=0$.
Then, for each $\epsilon > 0$, there is a 
$\delta = \delta(\epsilon,\psi,d) > 0$ 
such that $h(X+\p X) > -\infty$, and,
\begin{equation*}
\sigma[X] < 1+ \delta,
\end{equation*}
implies that  $d_{\rm LP}({\nu}, \gamma_{\mu,\Sigma}) < \epsilon$ 
for some for some $\mu,\Sigma$.
\end{theorem}


\subsection{The Gaussian as worst-case noise}
\label{s:gaussian}
Suppose $Y,Z$ are independent random variables, 
where $Z$ is a
Gaussian with the same variance as~$Y$. 
The maximum entropy property of the Gaussian 
states that $h(Z)\geq h(Y)$.
A natural but naive guess at a generalization 
of this property could be that,
for any random variable $X$ independent of $Y,Z$,
we always have:
\begin{equation} 
\label{gaussmixture}
h(X + Z) \geq h(X + Y).
\end{equation} 
As it turns out, this is not always true.
This fact has been part of the information-theoretic
folklore for quite some time.
An explicit counterexample was provided 
by Eskenazis, Nayar and Tkocz (ENT)~\cite{eskenazis:18},
and another counterexample was found by 
Madiman, Nayar, and Tkocz~\cite{madiman:20}, where $X
$ and $Y$ are in fact log-concave, 
symmetric and with unit variance. 

But there are also positive results,
along several different directions.
First, one can restrict the class
of distributions considered for $X$ and $Y$.
Indeed,
ENT~\cite{eskenazis:18} 
show that~(\ref{gaussmixture})
does actually hold for Gaussian scale mixtures, that is, for 
independent random variables 
$X$ and $Y$ that can be expressed 
as $\sqrt{W}Z'$, where $Z'\sim N(0,1)$ and $W$ is almost 
surely positive and independent of~$Z'$.

Second, one might ask for approximate 
versions of~(\ref{gaussmixture}).
Zamir and Erez~\cite{zamir:04}, motivated 
by the question of robustness of Gaussian inputs in additive-noise 
channels, showed that we always have:
\begin{equation} 
\label{ZE04result}
h(X + Z) \geq h(X + Y) - \frac{1}{2}\log{2}.
\end{equation}
Hence, even though~\eqref{gaussmixture} is not true in general, 
$h(X+Z)$ can not be much smaller than $h(X+Y)$. 

Third, we might restrict $Y$ to have the same distribution as $X$.
The question of 
whether~(\ref{gaussmixture}) holds under 
this condition was raised in ENT~\cite{eskenazis:18},
and it remains open. But
in this case we can establish that an improvement of 
bound~(\ref{ZE04result}) holds.
Specifically, in Theorem~\ref{thm:ENTqn}
we prove that, whenever $X,Y$ are i.i.d.\ random variables,
$$h(X+Z)\geq h(X+Y) -\frac{1}{2}\log\Big(\frac{2}{\varphi}\Big),$$
where $\varphi$ again denotes the golden ratio.
This improves the $\frac{1}{2}\log 2\approx 0.347$ error
term in~(\ref{ZE04result}) to $\frac{1}{2}\log(2/\varphi)\approx 0.106$.

\subsection{Robustness of Gaussian channel codes}

A first connection between entropic doubling
and Gaussian channel inputs
is revealed by re-interpreting our Theorem~\ref{prop:doub}.
As noted in Corollary~\ref{cor:wc} in Section~\ref{s:robust}, 
it can be rewritten as the following channel
capacity bound: If $Y=X+Z$ is an additive noise channel,
with input $X$ constrained to average power $P$,
and with noise $Z$ with power $N$, then,
$$C(Z;P)\geq I(X^*;X^*+Z)\geq \frac{1}{2}\log 
\Big(1+\frac{P}{N}\sigma[Z]\Big),$$
where
$C(Z;P)$ is the capacity of the channel,
$X^*\sim N(0,P)$ is a Gaussian input,
and $\sigma[Z]$ is the doubling constant of the noise $Z$.
Since $\sigma[Z]\geq 1$ with equality if and only if $Z$ is Gaussian,
this provides a pleasing proof of the worst-case noise property 
of the Gaussian distribution~\cite{cover:book2,lapidoth:96}.

The study of the robustness of Gaussian codebooks 
for additive noise channels with a power constraint
has a rich history, part of which is reviewed in 
Sections~\ref{s:robust} and~\ref{s:biblio}. 
Notably, Zamir and Erez~\cite{zamir:04}
showed that Gaussian inputs always come within
1/2 a bit of channel capacity:
$$I(X^*;X^*+Z)\geq C(Z;P)-\frac{1}{2}\log 2.$$
Since this bound is trivial when $C(Z;P)\leq \frac{1}{2}\log 2$,
they also raised the question of whether a
``multiplicative'' version of the above ``additive''
bound could be obtained, that would also provide
useful information in the ``low SNR'' regime, that is,
when the capacity is small.

Using an argument similar to that
employed in the proof of Theorem~\ref{prop:doub},
in Theorem~\ref{zamirimprovement} we offer a partial answer
to this question, by showing that, indeed,
we always have
$$I(X^*;X^*+Z)\geq \frac{\snr}{3\snr +2}C(Z;P),$$
where $\snr=\frac{P}{N}$.

Finally, in Sections~\ref{s:MAC} and~\ref{mimosection} we describe how the
same technique leads to analogous results on
the performance of time-shared Gaussian inputs 
for additive-noise multiple access channels,
and on the performance of uncorrelated Gaussian 
inputs for linear, additive-noise 
multiple-input multiple-output (MIMO) channels, respectively.

\subsection{Ideas and techniques}

The technical tools employed in our derivations are 
a combination of elementary techniques, probability-theoretic
arguments, and applications of
recently developed entropy inequalities whose implications 
for communication have perhaps 
not been fully tapped yet. The first such inequality -- the submodularity 
of the
entropy of sums~\cite{KM:14} -- has seen numerous applications in 
convex geometry, e.g.~\cite{bobkov:13,fradelizi:24b},
and in additive combinatorics, 
e.g.~\cite{MMT:12,kontoyiannis-MM-ISIT:10}. The second --
the fractional superadditivity of entropy 
power~\cite{artstein:04,madiman:19} --
has found applications in information 
theory, e.g.~\cite{KYBMadiman:12,ekrem:08,he:11,tang:11}.

In the context of the analogies between information theory and convex 
geometry outlined in~\cite{dembo-cover-thomas:91,madiman:17}, 
we note that there exist 
volume analogues
of the submodularity of entropy stated in
Lemma~\ref{lem:submod}, 
see~\cite{fradelizi:24b,fradelizi:24},
and of the fractional 
superadditivity of entropy power
stated in Lemma~\ref{lem:ep-fsa},
see~\cite{bobkov:11b,fradelizi:18}.
But these correspondences are limited; for instance, 
the statements for volumes sometimes require 
different constants. These limitations, and
especially the unavailability of an additive functional such as the 
variance in the geometric setting, make it difficult to write down 
analogues of our results there. 
Nonetheless, it is a somewhat curious coincidence that 
the golden ratio appears both in our Theorem~\ref{thm:ENTqn} 
and in~\cite{fradelizi:24b}.

\newpage

\section{Preliminaries}
\label{sec:prelim}

We begin by collecting some preliminary definitions
and results that will be used later on, and we also 
generalize a one-dimensional classical Gaussian stability 
result due to L\'{e}vy to the case of random vectors
in $\RL^d$.

Let $X, Y$ be random vectors taking values in $\R^d$, 
with probability density functions $f$ and $u$ respectively. 
The {\it relative entropy} is, as usual, defined by,
$$
D(X\|Y)=D(f\|u):=\int_{\R^d} f(x) \log\frac{f(x)}{u(x)} dx ,
$$
and is always nonnegative though possibly $+\infty$.
Furthermore, for a random vector $X$ with density $f$
and finite second moment we define,
$$
 D(X)=D(f):= \inf_{u} D(f\|u) = D(f\|g) ,
$$
with the infimum being taken over all Gaussian densities $u$ 
on $\R^d$, and achieved by the Gaussian density $g$ with the 
same mean and covariance matrix as $X$. The quantity $D(X)$ may 
be thought of as the {\it relative entropy from Gaussianity}. 
It is easy to see that
$D(f)=D(f\|g)=h(g)-h(f)$, which implies the Gaussian is the unique 
maximizer of entropy when the mean and covariance matrix are fixed.

Logarithms that appear in any expression involving entropies 
or relative entropies may be taken with respect to an arbitrary base, 
with the understanding that all entropies in that expression also use the 
same base. For concreteness, we assume that the base is~$e$
everywhere.

\medskip

\noindent
{\bf Standing assumption. } Throughout this paper, 
in any statement involving
entropy, entropy power, or a doubling constant, we implicitly
assume that those {\em exist} for the random variables
to which they correspond.

\medskip

We need two lemmas describing fundamental properties of the entropy 
under convolution. The first expresses the submodularity of the 
entropy of sums.

\begin{lemma} {\em \cite{KM:14}}
\label{lem:submod}
If $X_1, X_2, X_3$ are independent random vectors in $\RL^d$, then:
$$
h(X_1+X_2+X_3)+h(X_1) \leq h(X_1+X_2)+h(X_1+X_3).
$$
\end{lemma}

A version of this inequality for random variables taking values 
in a discrete group 
can be traced back to~\cite{KV:83}; it extends in fact to 
all locally compact groups~\cite{KM:16}.

The second lemma we need is the fractional 
superadditivity of the entropy power of 
convolutions.
We use $[n]$ to 
denote $\{1, 2,\ldots, n\}$.

\begin{lemma}
{\em \cite{artstein:04,madiman:19}}
\label{lem:ep-fsa}
$(i)$~If $X_1,X_2, \ldots, X_n$ are independent random vectors in $\RL^d$, 
then for any $1\leq k\leq n$:
\begin{equation}\label{mb-epi}
N(X_{1}+X_2+\cdots+X_{n}) \geq \frac{1}{\binom{n-1}{k-1}} 
\sum_{S\subset [n]:|S|=k} N \Big( \sum_{j\in S} X_{j}\Big) .
\end{equation}
$(ii)$ In particular, for $k=n-1$:
\begin{equation}\label{abbn-epi}
N(X_{1}+X_2+\cdots+X_{n}) \geq \frac{1}{n-1} \sum_{i\in [n]} 
N \Big( \sum_{j\neq i} X_{j}\Big) .
\end{equation}
\end{lemma}

The inequality~\eqref{abbn-epi} was proved 
in~\cite{artstein:04};
simpler proofs of it were later given 
in~\cite{madiman:07,tulino:06,shlyakhtenko:07}. 
The more general inequality~\eqref{mb-epi} 
was obtained in~\cite{madiman:19}. 

We will also need a classical fact about the well-known theorem 
of Cram{\'e}r~\cite{cramer:36}, which states that if the sum of two 
independent random vectors is Gaussian, then the random vectors 
themselves are 
also Gaussian. Let $d_{\rm LP}(\mu,\nu)$ denote 
the L{\'e}vy--Prokhorov metric between 
two probability measures $\mu,\nu$ on $\mathbb{R}^d$, 
$$
d_{\rm LP}(\mu,\nu) := \inf{\Big\{\epsilon>0:\nu(A) \leq \mu(A^{\epsilon}) 
+ \epsilon \text{ and } \mu(A) \leq \nu(A^{\epsilon})+\epsilon, \,
\text{for all Borel sets } A\subset\RL^d \Big\}},$$
where $A^{\epsilon}$ denotes the ``$\epsilon$-blow-up'' of $A$,
$$A^\epsilon
:=\Big\{x \in \mathbb{R}^d: \inf_{y\in A} \|x-y\| < \epsilon\Big\}.$$
When $\mu,\nu$ are the laws of random vectors, $X,Y$, respectively,
we find it convenient to write $d_{\rm LP}(X,Y)=d_{\rm LP}(\mu,\nu)$.
We will only make use of the fact that the L{\'e}vy--Prokhorov distance 
metrizes weak convergence in $\mathbb{R}^d$.

L{\'e}vy~\cite{levy:book} proved weak stability 
of Cram{\'e}r's theorem in dimension~$1$. 
Next, we prove the corresponding stability result
in $\RL^d$,
using a simple modification of 
L{\'e}vy's original argument.

\begin{theorem}[Multidimensional stability of Cram\'{e}r's theorem]
\label{levystabilityd}
Let $Z$ be a Gaussian random vector in $\RL^d$
with invertible covariance matrix $\Sigma$. 
For any $\epsilon>0$ there is a
$\delta=\delta(\Sigma,\epsilon)>0$ such that,
if the independent random vectors $X,Y$ have
$d_{\rm LP}(X+Y,Z)<\delta,$ 
then we also have $d_{\rm LP} (X,U_1) < \epsilon$ 
and $d_{\rm LP}(Y,U_2) < \epsilon$ 
for some Gaussian random vectors $U_1, U_2$.
\end{theorem}

\begin{proof}
Let $\{(X_n,Y_n)\;;\;n\geq 1\}$ be a sequence of pairs of 
random vectors with each $X_n$ independent of~$Y_n$,
and let $Z_n=X_n+Y_n$, $n\geq 1.$
Suppose that 
$\{Z_n\}$ converges weakly to $Z$ as $n\to\infty$,
but there is some $c >0$ such that, for every $n \geq 1$ 
and any Gaussian random vector $U$, $d_{\mathrm{LP}}(X_n,U) \geq c $. 
We will obtain a contradiction.

We assume without loss of generality that the (coordinatewise) 
median of each $X_n$ is zero. 
Let $\eta\in(0,1/2)$ and choose $\ell>0$ 
such that,
$${\mathbb P}\big(Z\in[-\ell,\ell]^d\big)=
\frac{1}{{\mathrm{det}(2 \pi\Sigma)}^{1/2}} 
\int_{[-\ell,\ell]^d}e^{-\frac{1}{2}x^T\Sigma^{-1}x} dx > 1- \eta.
$$
By weak convergence of the $Z_n$, this implies that 
the maximum concentration of $Z_n$ in a cube of side $2\ell$ exceeds 
$1-\eta$  eventually. Therefore, the same must be 
true for the maximum 
concentration cube of side $2\ell$ of $X_n$, since,
eventually,
\begin{align*}
1-\eta
&<
	\sup_{x \in \mathbb{R}^d}
	\mathbb{P}\Big(Z_n \in x + [-\ell,\ell]^d\Big) \\
&=
	\sup_{x \in \mathbb{R}^d}\mathbb{E}
	\left[
	\mathbb{P}\Big(
	X_n \in x - Y_n+[-\ell,\ell]^d\Big|Y_n \Big)
	\right]\\
&\leq 
	\sup_{z \in \mathbb{R}^d}
	\mathbb{P}\Big(X_n \in z + [-\ell,\ell]^d\Big).
\end{align*}
Since $\eta < 1/2$, there must be a $z$ such that
this bound is satisfied and 
$z+[-\ell,\ell]^d$ includes the origin,
otherwise the median of some coordinate would be nonzero. 
Therefore, for any $\epsilon>0$,
$$
{\mathbb P}\big(\|X_n\| > 2\ell\sqrt{d}\big) < \epsilon,
$$
eventually. 
The laws of $\{X_n\}$ thus form a compact set with respect to weak topology. 

By compactness, there is a convergent subsequence $\{X_{n_k}\}$. 
Since by assumption we have that $d_{\rm LP}(X_{n_k},U) > c > 0$ 
for any Gaussian $U$, and since $d_{\rm LP}(\cdot,\cdot)$ metrizes weak 
convergence, the weak limit, say $V$, of that subsequence is not Gaussian.
Finally, since we have convergence 
of $X_n$ and $Z_n$ along this subsequence and $X_n$ is independent 
of $Y_n$, $Y_n$ must also converge along the same subsequence to some 
random vector $W$. But we have thus obtained two independent random 
vectors $V,W,$ whose sum, $Z$, is Gaussian without $V$ being Gaussian. 
This contradicts the multivariate version of Cram{\'e}r's 
theorem~\cite{cramer:36}.
\end{proof}

Let us remark here that the recent work~\cite{mahvari:23} also 
considers a stability theorem of another characterization of the normal 
distribution, namely Bernstein's theorem, which says that if the sum and 
difference of two independent random variables are independent, then the 
random variables are Gaussian. That stability result is 
also generalized to higher dimensions in~\cite{mahvari:23} and it is used 
to characterize capacity regions for certain additive Gaussian noise 
channels.

Finally, we need a qualitative stability result for the 
EPI that was obtained by Carlen and 
Soffer~\cite{carlen:91} under a finite Fisher information assumption. 
In order to state it in a form more suited 
to our purposes, we need some notation. 

For a random vector $X$ with an absolutely continuous
density function $f$, 
the {\it Fisher information} of $X$ is defined as:
$$
I(X)=I(f):=\int_{\mathbb{R}^d} 
\bigg[\frac{\|\nabla f(x)\|}{f(x)}\bigg]^2 f(x) dx .
$$
Note that the derivative is well defined almost everywhere as a consequence of the absolute continuity of $f$; if this condition is not satisfied, we set $I(X)=\infty$.
We can now define another useful class of densities:
For a given $I_0>0$ and 
a given 
function $\psi:[0,\infty)\ra [0,\infty)$ with $\lim_{R\ra \infty} \psi(R)=0$, 
let:
\begin{equation}
\mathcal{C}_{I_0, \psi}:=\big\{ f\in \mathcal{F}_d: 
I(f)\leq I_0, \psi_f \leq \psi\big\} .
\label{eq:CI0}
\end{equation}
Observe that $\mathcal{C}_{I_0, \psi}$ is a 
subset of ${\mathcal C}_\psi$ and a strict subset 
of the densities with variance at most $\psi(0)$.
The following is a slight rephrasing of~\cite[Theorem~1.2]{carlen:91}.

\begin{proposition}{\em \cite{carlen:91}}
\label{prop:carlen}
For any $I_0>0$, $\epsilon>0$, $a\in(0,1/2]$,
and any $\psi:[0,\infty)\ra [0,\infty)$ with $\lim_{R\ra \infty} \psi(R)=0$, 
there is a $\delta=\delta(\epsilon,I_0,\psi,a) > 0$ such that, if:
\begin{enumerate}
    \item $X_1$ and $X_2$ are independent,
	with densities in $\mathcal{C}_{I_0,\psi}$;
    \item both $X_1$ and $X_2$ have zero mean and identity covariance;
    \item $\lambda\in [a,1-a]$ for some $a\in (0,\frac{1}{2}]$; 
    \item and $h(\sqrt{\lambda}X_1+\sqrt{1-\lambda}X_2)-
	\big[\lambda h(X_1)+(1-\lambda) h(X_2)\big]
	<  \delta,$
\end{enumerate}  
then
$D(X_1) < \epsilon$ and $D(X_2) < \epsilon$.
\end{proposition}

\newpage

\section{Entropy comparison inequalities}
\label{s:main}

In this section we prove the entropy bounds 
discussed 
in Sections~\ref{s:smalllarge} and~\ref{s:gaussian}
of the Introduction.

\subsection{Adding Gaussian noise}
\label{sec:comp}

We first examine the $d$-dimensional 
version of one of the main problems
described in the Introduction, namely,
comparing the entropy $h(X+X')$ of the sum
of two i.i.d.\ random vectors with the entropy
$h(X+Z)$ obtained when one of the random vectors
is replaced by a Gaussian with the same covariance.
The question of the relationship between
$h(X+X')$ and $h(X+Z)$ is of course a natural one,
perhaps first explicitly discussed 
by ENT~\cite{eskenazis:18} in the one-dimensional case.
While the validity of the inequality
$h(X+Z)\geq h(X+X')$ suggested in~\cite{eskenazis:18} remains open, 
our first main result below
establishes an approximate version
for $d$-dimensional random vectors,
with an explicit
error term.
For i.i.d.\ random vectors,
this improves the $\frac{1}{2}\log 2$ bound~(\ref{ZE04result})
established by Zamir and Erez~\cite{zamir:04} in dimension~1,
and it is the best bound known to date.

\begin{theorem}[Approximate generalized Gaussian maximum entropy]
\label{thm:ENTqn}
Let $X,X'$ be i.i.d.\ random vectors in $\RL^d$, and $Z$ 
be a Gaussian random vector with the same covariance matrix as $X$.
Then,
\begin{equation*}
h(X+Z) \geq h(X+\p{X}) -\frac{d}{2}\log{\Big(\frac{2}{\varphi}\Big)}.
\end{equation*}
where $\varphi=\frac{1+\sqrt{5}}{2}>1$ is the golden ratio.
\end{theorem}

\begin{proof} 
The main idea of the proof is to bound $h(X + Z + \p{X})$ 
from above using Lemma~\ref{lem:submod} and from below 
using Lemma~\ref{lem:ep-fsa}~$(ii)$.
Thus, we obtain,
\begin{align*} 
    2h(X+Z) - h(Z) &\geq h(X+Z+\p{X})
    \\
    &\geq h(X+\p{X}) 
    + \frac{d}{2}\log{\Bigl[1+ 2e^{\frac{2}{d}(h(X+Z)-h(X+\p{X}))}}\Bigr]
    - \frac{d}{2}\log{2},
\end{align*}
or,
\begin{align} 
    h(X+Z) 
&
    \;\;\geq h(X+\p{X}) + \frac{d}{4}\log{\Bigl[1+ 2e^{\frac{2}{d}(h(X+Z)-h(X+\p{X}))}}\Bigr]
    \nonumber\\
&
    \;\;\;\;\;\;
    -\frac{1}{2}\bigl[h(X+\p{X}) - h(Z)\bigr] - \frac{d}{4}\log{2}
    \nonumber\\
&
    \;\; \geq  h(X+\p{X}) + \frac{d}{4}\log{\Bigl[1+ 2e^{\frac{2}{d}(h(X+Z)-h(X+\p{X}))}}\Bigr]
    - \frac{d}{2}\log{2};
    \label{lb3}
\end{align}
cf.~equation~(\ref{eq:intro}) in the Introduction.
Consider the increasing function,
$$g(a) := a - \frac{d}{4}\log{\bigl[1+2e^{\frac{2}{d}a}\bigr]},
\quad a\in\RL.$$
Using the fact that its inverse is, 
$$g^{-1}(y) = \frac{d}{2}\log{\Bigl[e^{\frac{4}{d}y} + e^{\frac{2}{d}y}\sqrt{e^{\frac{4}{d}y} + 1}\Bigr]},$$
and minimizing the quadratic in
the argument of the logarithm, 
from~\eqref{lb3} we obtain,
\begin{align*}
h(X+Z) - h(X+\p{X})
& \geq 
    g^{-1}\Big(-\frac{d}{2}\log{2}\Big)\\
& =
    -\frac{d}{2}\log{\Big(\frac{2}{\varphi}\Big)},
\end{align*}
completing the proof.
\end{proof} 

\subsection{Large doubling}
\label{sec:l-doub}

Here we obtain another 
quantitative consequence of large doubling
for the entropy.  Theorem~\ref{thm:ENTqn} gives 
one such result --
indeed, it can be rewritten as:
\begin{equation}
h(X+Z)-h(X) \geq  \frac{d}{2}
\log{\big(\varphi \sigma[X] \big)}.
\label{eq:bigdouble}
\end{equation}
The following theorem gives an alternative
inequality of the form~(\ref{eq:bigdouble}),
with the additional 
advantage that it is sharp.

\begin{theorem}[Large doubling]
\label{prop:doub}
Suppose $X$ is a random vector in $\R^d$ with 
covariance matrix $K_X$, 
and $Z$ is an independent Gaussian random vector with covariance matrix $K_Z$. 
Then,
\begin{equation} 
\label{doublingbound}
h(X+Z)-h(X) \geq  \frac{d}{2}\log{\big(1 + a \sigma[X]\big)},
\end{equation}
where $a^d=\frac{\det(K_Z)}{\det(K_X)}$. Moreover, equality holds 
in \eqref{doublingbound} if $X$ is Gaussian with  $K_X=aK_Z$ for some $a>0$.
\end{theorem}

\begin{proof}
By the multivariate 
central limit theorem and the upper semi-continuity of entropy 
with respect to weak convergence, 
\begin{equation} 
\label{upeersemi}
h(X+Z) \geq \limsup_{n\to\infty}{h\left(X_0 + (K_Z^{1/2}K_X^{-1/2})
\frac{1}{\sqrt{n}}\sum_{i=1}^n X_i\right)},
\end{equation}
where $\{X_n\;;\;n\geq 0\}$ are i.i.d.\ copies of $X$. 
For fixed $n \geq 1$, the scaling property of entropy power implies that,
\begin{align}
N\Bigg(
&
	X_0 + (K_Z^{1/2}K_X^{-1/2})\frac{1}{\sqrt{n}}
	\sum_{i=1}^n X_i \Bigg) 
	\nonumber\\
&= 
	\frac{1}{n} N\Bigg(\sqrt{n}X_0 
	+ \sum_{i\in [n]} K_Z^{1/2}K_X^{-1/2} X_i\Bigg)
	\nonumber\\ 
&\geq 
	\frac{1}{n^2}\bigg[\sum_{i\in [n]} 
	N\big(\sqrt{n}X_0 + K_Z^{1/2}K_X^{-1/2} X_i\big) 
	+ \sum_{i\in [n]:i<j} N\big(K_Z^{1/2}K_X^{-1/2} (X_i+X_j)\big) \bigg]
	\label{genEPI}
	\\
&=
	\frac{1}{n^2}\bigg[ n 
	N\big(\sqrt{n}X_0 + K_Z^{1/2}K_X^{-1/2} X_1\big)
	+\binom{n}{2} N\big(K_Z^{1/2}K_X^{-1/2}(X_1+X_2)\big) \bigg] 
	\nonumber\\ 
&\geq 
	\frac{1}{n} N(\sqrt{n}X_0) 
	+ \Big(\frac{n-1}{2n}\Big) N(K_Z^{1/2}K_X^{-1/2}(X_1+X_2))
	\label{conv}
	\\ 
&= 
	N(X_0) + \Big(\frac{n-1}{2n}\Big) a N(X_1 + X_2) ,
	\label{last}
\end{align}
where in~\eqref{genEPI} we used
Lemma~\ref{lem:ep-fsa}~$(i)$ with $k=2$,
and~\eqref{conv} follows 
from the fact that entropy does not decrease on convolution. 
Combining~\eqref{last} and~\eqref{upeersemi} yields the claimed 
bound~(\ref{doublingbound}). 

The sufficient condition for equality follows by direct computation.
\end{proof}

Since we have equality in~(\ref{doublingbound}) if $X$ is Gaussian,
and because of the use of the EPI in the proof, it might be tempting 
to conjecture that the corresponding `only if' statement is also true, 
that is, that equality is achieved only if $X$ is Gaussian with 
$K_X$ proportional to $K_Z$.
However, this is not the case, as the following example illustrates:
Let $X$ be a random variable 
such that $h(X) = -\infty$, but $h(X+\p{X}) > -\infty$;
see, e.g.,~\cite{bobkov:15} for an explicit such construction. 
Then $X$ is necessarily not Gaussian, while
both sides of~\eqref{doublingbound} 
are~$+\infty$. 
As we are not aware of a corresponding
example with $h(X)>-\infty$, we
do not know whether an `only if' statement 
for the case of equality may be 
proved under the assumption $h(X)>-\infty$ or, more generally,
under other additional assumptions on $X$.






\subsection{Small doubling}
\label{sec:s-doub}

As observed in the Introduction,
$D(X)=h(Z)-h(X)\geq \frac{d}{2} \log\sigma[X]$. 
Therefore, $\sigma[X]$ being
bounded away from~1 implies that $X$ is bounded away 
from being Gaussian in the sense of relative entropy. 
It is a natural to ask if the reverse holds -- 
does $\sigma[X]$ close to~1 imply 
that $X$ is close to Gaussian? 
Courtade, Fathi, and 
Pananjady~\cite{courtade:18} answer
this stability question by presenting a 
counterexample:
They construct
a sequence $\{X_n\}$ of random variables with 
$\sigma[X_n]\to 1$ as $n\to\infty$, but whose 
quadratic Wasserstein distance from
the Gaussian is bounded 
away from~0.

Nevertheless, stability can hold for restricted
classes of distributions.  The qualitative 
stability result of~\cite{carlen:91} under a bounded 
Fisher information condition was the first such result.
A sharp, quantitative stability result was 
obtained in~\cite{KM:14} under the finite Poincar{\'e} 
constant assumption.
And more recently, quantitative stability results have 
been obtained under regularity assumptions such as 
log-concavity~\cite{courtade:18,eldan:20}.

Our main observation in this section is that, although stability of 
the lower bound on doubling may fail for strong distances, one can 
obtain a simpler qualitative characterisation of random variables 
with small doubling in terms of weak convergence, 
by combining observations from~\cite{carlen:91} with 
the stability of Cram\'{e}r's theorem due to L{\'e}vy. 
Recall the definitions of the classes of densities
$\mathcal{C}_{\psi}$ and $\mathcal{C}_{I_0, \psi}$ 
in~(\ref{eq:Cpsi}) and~(\ref{eq:CI0}), respectively.

\begin{theorem}[Small doubling implies near-Gaussianity]
\label{smalldoublinglevy2}
Let $X_1,X_2$ be independent random vectors in $\RL^d$ with laws 
$\nu_1,\nu_2$, respectively, with densities 
in $\mathcal{C}_{\psi}$
for some $\psi:[0,\infty)\ra [0,\infty)$ with 
$\lim_{R\ra \infty} \psi(R)=0$.
Suppose $X_1,X_2$ have zero mean,
the same invertible covariance matrix $\Sigma$,
and $h(X_1+X_2) > -\infty$.  Then for 
each $\epsilon > 0$, there is a $\delta = \delta(\epsilon,\psi, \Sigma,d)>0$
such that,
\begin{equation*} 
h\bigg(\frac{X_1 + X_2}{\sqrt{2}}\bigg) -\frac{h(X_1)+h(X_2)}{2} <  \delta,
\end{equation*}
implies that  $d_{\rm LP}({\nu_1}, \gamma_{\mu_1,\Sigma_1}) 
< \epsilon$ and $d_{\rm LP}({\nu_2}, \gamma_{\mu_2,\Sigma_2}) < \epsilon$
for some for some $\mu_1,\Sigma_1,\mu_2,\Sigma_2$.
\end{theorem}


\begin{proof}
Fix $\epsilon > 0$.
First,  assume 
that $\Sigma = \mathrm{I}_d$. 

For $i = 1,2$ 
let $\tilde{X}_i = \frac{X_i + Z_i}{\sqrt{2}},$ where $Z_1,Z_2$ 
are independent standard Gaussian random vectors.
We have the elementary estimate~\cite[Eq.~(1.38)]{carlen:91},
\begin{equation} \label{elementaryestimate}
\psi_{\tilde{X}_i}(R) \leq 2\psi_{X_i}\bigl(R/2\bigr) 
+ 2\psi_{Z_i}\bigl(R/2\bigr).
\end{equation}
Furthermore, 
\begin{equation}
I(\tilde{X}_i) \leq 2I(Z_i) = 2d,
\end{equation}
and by~\cite[Eq.~(1.30)]{carlen:91} we also have,
\begin{equation} \label{smoothdoubling}
h(\tilde{X}_1 + \tilde{X}_2) -\frac{h(\tilde{X}_1)+h(\tilde{X}_2)}{2} \leq h(X_1 + X_2) -\frac{h(X_1)+h(X_2)}{2}.
\end{equation}
Let $\epsilon_1 > 0$ to be chosen later. 
In view of~\eqref{elementaryestimate}--\eqref{smoothdoubling}, 
by Proposition~\ref{prop:carlen}
with $\lambda = 1/2$, there is a $\delta(\epsilon_1,\psi,d) > 0$ such that,
\begin{equation} \label{scaleinv}
h(X_1 + X_2) -\frac{h(X_1)+h(X_2)}{2} < \frac{1}{2}\log{2} + \delta \quad {\rm implies} \quad D(X_i+Z_i) = D(\tilde{X}_i) < \epsilon_1,
\end{equation}
for $i = 1,2.$ By Pinsker's inequality, 
$\|\nu_i * \gamma_{0,\mathrm{I}_d} 
-  \gamma_{0,2\mathrm{I}_d}\|_{\rm TV} 
\leq \sqrt{\frac{\epsilon_1}{2}}$,
where `$*$' denotes the convolution of
probability measures. Thus,
$$
d_{\rm LP}(\nu_1* \gamma_{0,\mathrm{I}_d},  \gamma_{0,2\mathrm{I}_d}) 
\leq \|\nu_1 * \gamma_{0,\mathrm{I}_d} -  \gamma_{0,2\mathrm{I}_d}\|_{\rm TV} 
\leq \sqrt{\frac{\epsilon_1}{2}},
$$
and,
$$d_{\rm LP}(\nu_2 * \gamma_{0,\mathrm{I}_d},  \gamma_{0,2\mathrm{I}_d}) 
\leq \|\nu_2 * \gamma_{0,\mathrm{I}_d} 
- \gamma_{0,2\mathrm{I}_d}\|_{\rm TV}\leq \sqrt{\frac{\epsilon_1}{2}}.
$$
By Theorem \ref{levystabilityd}, 
$$
d_{\rm LP}(\nu_1, \gamma_{\mu_1,\Sigma_1}) < \epsilon \quad {\rm and} \quad 
d_{\rm LP}(\nu_2, \gamma_{\mu_2,\Sigma_2}) < \epsilon,
$$
for some $\mu_1,\Sigma_1,\mu_2,\Sigma_2$ and an appropriate choice 
of $\epsilon_1$. 
Since $\epsilon > 0$ was arbitrary, this gives the claimed
result.

Finally, to see that the assumption
$\Sigma = \mathrm{I}_d$
incurs no loss of generality,
note that otherwise the 
assumption of the theorem holds for the 
random variables $\Sigma^{-\frac{1}{2}}X_i$ and~\eqref{scaleinv} 
is scale-invariant. 
\end{proof}

\noindent
{\bf Remarks.}
\begin{enumerate}
\item
For identically distributed random vectors, the above result provides a characterization of the distribution of random vectors with small doubling. For random variables with the same variance (but not necessarily independent), it may also be seen as a 
characterization of random variables with {\em entropic Ruzsa 
distance}~\cite{KM:14,tao:10} close to $\frac{1}{2}\log{2}.$
\item
The problem of stability in Cram{\'e}r's theorem is well studied for $d=1$. 
Under the assumption that the random variables have median zero, 
stability results for the Kolmogorov distance are 
known~\cite{sapogov:55} with no moment 
assumptions. Furthermore, assuming $\Var{(X_1)} = \Var{(X_2)}$ sharp rates 
were obtained recently~\cite{bobkov:13b,bobkov:16,bobkov:16b}, 
improving earlier results of Sapogov~\cite{sapogov:55}. The book 
of Linnik and Ostrovski{\u\i}~\cite{linnik:book} contains further discussion 
and references. 
\end{enumerate}

\newpage

\section{Channel coding bounds}
\label{sec:channel}

\subsection{Robustness of Gaussian inputs}
\label{s:robust}

The channel coding theorem for memoryless channels states
that a channel's {\em capacity}, defined as the fastest possible 
reliable
communication rate, is equal to the maximum of 
the mutual information $I(X;Y)=h(X)-h(X|Y)$ 
between the channel input $X$ and the channel output $Y$,
where the
maximum is over all appropriate input distributions. 
Moreover, the {\em capacity-achieving distribution} 
determines the distribution of optimal random codebooks 
and it informs the construction 
of good codes~\cite{shamai:97,han:93,polyanskiy:14}.
The capacity-achieving distribution of the
the Additive White Gaussian Noise (AWGN) channel 
with an average power constraint
is itself Gaussian~\cite{shannon:48},
but in general the determination of capacity-achieving 
distributions is a fraught task;
see, e.g.,~\cite{tchamkerten:04,huang:07,dytso:20},
for discussions of the
associated difficulties.

Gaussian codebooks occupy a special place
in the study of additive-noise channels.
For example, 
i.i.d.\ Gaussian inputs and i.i.d.\ Gaussian noise 
are minimax optimal choices in the channel coding game~\cite{cover:book2}.
Moreover, for any noise distribution with power~$N$,
Gaussian codebooks with power $P$ can achieve all rates 
up to the Gaussian channel
capacity, $\frac{1}{2}\log(1+\frac{P}{N})$~\cite{lapidoth:96,scarlett:17}.

The robustness of Gaussian codebooks was further
examined by Zamir and Erez~\cite{zamir:04},
who showed that,
for {\em any} memoryless, additive-noise channel 
with an average power constraint,
Gaussian inputs incur a rate loss no greater
that half a bit compared to capacity.
Specifically, consider a memoryless
additive-noise channel
with input
$X\in\RL$, noise $Z\in\RL$, and output
$Y=X+Z$.
Let $N=\Var(Z)>0$ denote
the noise power and $P>0$ the input power constraint,
that is, suppose that all potential input codewords
$(x_1,x_2,\ldots,x_n)\in\RL^n$ are constrained
as,
$$
\frac{1}{n}\sum_{i=1}^n x_i^2 \leq P.
$$
Let $C(Z;P)$ be the capacity of this channel:
$$C(Z;P)
=
    \sup_{X: {\mathbb E}(X^2)\leq P} I(X;X+Z).
$$

First we note that Theorem~\ref{prop:doub}
established earlier, immediately implies the following 
simple bound on the performance of Gaussian inputs
for an arbitrary power-constrained additive-noise channel.

\begin{corollary}[Quantitative minimax optimality of Gaussian noise]
\label{cor:wc}
Let $X^*\sim N(0,P)$. The capacity
$C(Z;P)$ of the channel with additive noise $Z$ 
with power $N=\Var(Z)$
and average power constraint $P$, satisfies,
\begin{equation}
C(Z;P)\geq I(X^*;X^*+Z)\geq \frac{1}{2}\log \big[1+\snr \,\sigma[Z]\big] ,
\label{eq:worst}
\end{equation}
where $\snr:=\frac{P}{N}$ denotes the signal-to-noise-ratio 
and $\sigma[Z]$ is the doubling constant of the noise $Z$.
\end{corollary}

Since $\sigma[Z]\geq 1$ with equality if and only if $Z$ is Gaussian,
from (\ref{eq:worst}) we immediately recover the worst-case property 
of Gaussian noise; see also~\cite{ihara:78,diggavi:01} for 
further developments. 

Zamir and Erez's complementary result 
on the robustness of Gaussian inputs reads as follows: 

\begin{theorem}
{\em \cite{zamir:04}}
\label{halfabit}
Let $X^*\sim N(0,P)$. The capacity
$C(Z;P)$ of the channel with additive noise $Z$
and average power constraint $P$, satisfies:
$$I(X^*; X^*+Z)\geq C(Z;P)- \frac{1}{2}\log 2 .$$
\end{theorem}

As this additive bound is clearly only nontrivial
if the capacity $C(Z;P)>\frac{1}{2}\log 2$,
Zamir and Erez~\cite{zamir:04}
raised the question of whether there exist
a constant $\alpha>0$ and
an input $X'$, possibly non-Gaussian, 
for which a multiplicative bound of the form,
\begin{equation} 
\label{zamirquestion}
I(X';X'+Z) > \alpha C(Z;P),
\end{equation}
might hold,
for any noise distribution for $Z$.
We offer a partial answer to this question
in Theorem~\ref{zamirimprovement}.
Generalizing the proof method of our 
earlier Theorem~\ref{prop:doub}, we show 
that a Gaussian input $X'=X^*$ indeed satisfies~(\ref{zamirquestion}),
but with a constant $\alpha$ 
which depends on the signal-to-noise ratio (SNR).
Theorem~\ref{zamirimprovement},
proved Section~\ref{s:proof},
gives meaningful information on the rate achieved by
Gaussian inputs at arbitrarily low SNR.

\begin{theorem}[Robustness of Gaussian inputs]
\label{zamirimprovement}
Let $X^*\sim N(0,P)$. The capacity
$C(Z;P)$ of the channel with additive noise $Z$ 
and $\snr=\frac{P}{N}$, where $P$ is the input
average power constraint and
$N=\Var(Z)$ is the noise power,
satisfies:
$$I(X^*;X^*+Z) \geq \frac{\snr}{3\snr+2} C(Z;P).$$
\end{theorem}

An analogous result for additive-noise
multiple access channels is established
in Theorem~\ref{MACimprovement} in Section~\ref{s:MAC},
and for linear additive-noise MIMO channels 
in Theorem~\ref{MIMOimprovement} in Section~\ref{mimosection}.
Although all three of these results
are stated for one-dimensional channels,
their extension to 
$\RL^d$ for any finite dimension $d\geq 1$ is
entirely straightforward.

\subsection{Bibliographical comments}
\label{s:biblio}

Before proceeding to prove 
Theorems~\ref{zamirimprovement}--\ref{MIMOimprovement},
we make
some more bibliographical remarks.

First, we note that 
further results on robust channel input distributions
in a similar spirit to~\cite{zamir:04}
are given, e.g., in~\cite{zamir:02,shulman:04,love:06,philosof:07},
and rate-distortion analogues to some of these bounds
are derived, e.g., in~\cite{feder-zamir:92,K-zamir:06}. 

In a different line of work,
{\it structured} random codes have  been 
shown to be both 
practical and effective; 
the survey~\cite{ramji:book} traces the story that 
began with~\cite{barron-joseph:12}. Moreover, 
although these ``sparse superposition'' codes are
built with Gaussian entries and were originally 
developed for the AWGN channel, 
they have subsequently been shown to be universal~\cite{barbier:19}. 
The aforementioned robustness results may 
be seen as approximate analogues of this for 
unstructured random codes.

In yet another direction,
the {\it restricted capacity}
of an additive-noise channel with noise $Z$ and input 
power constraint $P$ is defined by,
$$
C^{LC}(Z;P):= \sup_{X} I(X;X+Z) ,
$$
where the supremum is taken over all  
symmetric and log-concave random variables
$X$ such that $\E(X^2) \leq P$. 
Then~\cite[Corollary~3]{madiman:21} shows that,
if the noise $Z$ is symmetric and log-concave, we have:
$$
0\leq C(Z;P)-C^{LC}(Z;P) \leq \frac{1}{2} \log\Big( \frac{\pi e}{6}\Big).
$$
In other words, if the noise is symmetric and log-concave,
then the rate loss with respect to the ``true'' channel capacity 
incurred by restricting the input 
to also be symmetric and log-concave, is at most 
$\frac{1}{2}\log_2 \big(\frac{\pi e}{6}\big)\approx 0.254$ bits. 

\subsection{Proof of Theorem~\ref{zamirimprovement}}
\label{s:proof}

Let $X$ be any input with variance $P_X \leq P$. 
Let $Z$ be the noise with variance $N$. Assume without 
loss of generality that $X$ and $Z$ have zero mean
and that $h(Z) > -\infty$. 

By the central limit theorem, $X^*$ can be expressed 
as the weak limit of,
$$
\frac{\sum_{i=1}^n{(X_i} +{Z_i})}{\sqrt{2n\frac{P_X+N}{P}}},
$$
where $\{X_i\}$ and $\{Z_i\}$ are i.i.d.\ copies of $X$ and $Z$,
respectively. By the upper-semicontinuity of entropy,
\begin{equation} \label{uppersemiMI}
h(X^*+Z) \geq \limsup_{n\to\infty}
{h\left( \frac{\sum_{i=1}^n{(X_i} +{Z_i})}{\sqrt{2n\frac{P_X+N}{P}}}
+Z \right)}.
\end{equation}
By the scaling property of differential entropy, we have, for any $n \geq 1$,
\begin{align} 
\nonumber
&h\left(\frac{X_1 + \cdots + X_n+Z_1 + \cdots + Z_n}
	{\sqrt{2n\frac{P_X+N}{P}}}+Z\right) \\ 
\label{scalingbeforeEPIMI}
&= h\left(
	 X_1 + \ldots + X_n+Z_1 + \ldots + Z_n
	+\sqrt{2n\frac{P_X+N}{P}}Z
	\right) - \frac{1}{2}\log{\Bigl({2n\frac{P_X+N}{P}}\Bigr)} 
\end{align}
Now we apply~\cite[Theorem~3]{madiman:07}, with $\mathcal{C}$ 
being the following 
collection of subsets of 
$\{1,\ldots,2n+1\}$ of cardinality $2$: 
$\big\{\{1,2\},\ldots,\{1,n+1\}\big\} 
\cup \big\{\{i,j\}:2\leq i \leq n, n + 1\leq j \leq 2n + 1\big\}$.
This gives:
\begin{align*}
 e^{2h\Bigl(
	X_1 + \cdots + X_n+Z_1 + \cdots + Z_n
	+
	\sqrt{2n\frac{P_X+N}{P}}Z 
	\Bigr)} 
&\geq 
\frac{1}{n}\Biggl[ne^{2h\Bigl(\sqrt{2n\frac{P_X+N}{P}}Z + Z_1\Bigr)}  
	+ {n^2}e^{2h(X_1+Z_1)}  \Biggr]  \\
&\geq e^{2h\Bigl(\sqrt{2n\frac{P_X+N}{P}}Z\Bigr)} + {n}e^{2h(X_1+Z_1)},
\end{align*}
where the last inequality follows the fact that entropy does not decrease 
on the addition of an independent random variable. Thus, 
by~\eqref{scalingbeforeEPIMI}, 
and scaling,
\begin{align} 
&
	h\Bigl(\frac{X_1 + \cdots + X_n+Z_1 + \ldots + Z_n}
	{\sqrt{2n\frac{P_X+N}{P}}}+Z\Bigr) 
	\nonumber\\ 
&\qquad\geq 
	\frac{1}{2}\log{\Biggl( e^{2h\Bigl(\sqrt{2n\frac{P_X+N}{P}}Z\Bigr)} 
	+ {n}e^{2h(X_1+Z_1)}\Biggr)} 
	- \frac{1}{2}\log{\Bigl({2n\frac{P_X+N}{P}}\Bigr)} 
	\nonumber\\
&\qquad=
	\frac{1}{2}\log{\Biggl(e^{2h(Z)} 
	+ \frac{P}{2(P_X+N)}e^{2h(X+Z)}\Biggr)} 
	\nonumber\\ 
&\qquad\geq 
	\frac{1}{2}\log{\Biggl(e^{2h(Z)} 
	+ \frac{P}{2(P+N)}e^{2h(X+Z)}\Biggr)},
	\nonumber
\end{align}
where in the last inequality we used that $P_X \leq P$.
Therefore, since
$I(X;X+Z)=h(X+Z)-h(Z)$, we obtain:
\begin{align}
&
	h\Bigl(\frac{X_1 + \cdots + X_n+Z_1 + \ldots + Z_n}
	{\sqrt{2n\frac{P_X+N}{P}}}+Z\Bigr) 
	\nonumber\\
&\qquad= 
	I(X;X+Z)
	+ \frac{1}{2}\log{\Biggl(e^{-2I(X;X+Z)} 
	+ \frac{P}{2(P+N)}\Biggr)} + h(Z).
	\label{mutualinfoout}
\end{align}

We have the elementary bound 
$$
\frac{1}{2}\log{\Bigl(e^{-2x} + \beta\Bigr)} \geq -\frac{1}{1+\beta}x, \quad x \geq 0, \beta \geq 0.
$$
Thus, by~\eqref{mutualinfoout},
\begin{align*}
h\left(\frac{X_1 + \cdots + X_n+Z_1 + \cdots + Z_n}
	{\sqrt{2n\frac{P_Y+N}{P}}}+Z\right)  
&\geq \Bigl(1-\frac{2}{2+\frac{P}{P+N}}\Bigr)I(X;X+Z) + h(Z) \\
&= \frac{\snr}{3\snr+2}I(X;X+Z) + h(Z),
\end{align*}
where the right-hand side is independent of $n$.
Therefore, by~\eqref{uppersemiMI},
$$
h(X^*+Z) \geq  \frac{\snr}{3\snr+2}I(X;X+Z) + h(Z),
$$
or,
\begin{equation} \label{eqsnr}
I(X^*;X^*+Z) \geq \frac{\snr}{3\snr+2}I(X;X+Z).
\end{equation}
Finally, 
taking the supremum over all $X$ with 
$\mathbb{E}(X^2) \leq P$, gives the result. 
\qed

\subsection{Multiple access channels}
\label{s:MAC}

Consider a memoryless, additive-noise,
multiple access channel (MAC) 
with two senders $X_1, X_2$,  
power constraints $P_1, P_2$,
and output $Y = X_1 + X_2 + Z$, 
where $Z$ is arbitrary real-valued 
additive noise of power $N=\Var(Z)$.
The capacity region $\mathcal{C}(Z;P_1,P_2)$ of
this channel is given~\cite{verdu:90} by
the closure of the set of all 
pairs $(R_1,R_2)$ of nonnegative rates
$R_1,R_2$ satisfying,
\begin{align} 
R_1 &\leq 
	I(X_1;Y|X_2,V)
	=\sum_v{\mathbb P}(V=v) I(X_1;Y|X_2,V=v),
    \label{r1}\\ 
R_2 &\leq I(X_2;Y|X_1,V)
	=\sum_v{\mathbb P}(V=v) I(X_2;Y|X_1,V=v),
    \label{r2}\\ 
R_1 + R_2 &\leq I(X_1,X_2;Y|V)
	= \sum_v{\mathbb P}(V=v) I(X_1,X_2;Y|V=v),
    \label{r1r2sum}
\end{align}
for some triplet $(V,X_1,X_2)$ of random variables 
with $\mathbb{E}(X^2_i) \leq P_i$ for $i = 1,2$,
with $X_1$ and $X_2$ being conditionally independent given $V$,
with $V$ taking at most $2$ values, and with $(V,X_1,X_2)$ being 
independent of the noise $Z$. 

In~\cite{philosof:07}, the sum-capacity 
of the MIMO-MAC channel is studied, and it is noted there
that the half-a-bit bound of Theorem~\ref{halfabit}
extends to the capacity region of the MAC. Here will 
apply our Theorem~\ref{zamirimprovement}, more 
specifically the bound established in~\eqref{eqsnr}, 
to obtain an improvement of this bound 
for the MAC capacity region in the low-SNR regime. 

\begin{theorem}[MAC: Robustness of Gaussian inputs]
\label{MACimprovement}
Let $\mathcal{C}(Z;P_1,P_2)$ 
denote the capacity region 
of the two-user,
additive-noise MAC with noise $Z$ having power $N$, 
and power constraints $P_1,P_2$. Consider the rate region
$\mathcal{C^*}(Z;P_1,P_2)$ defined as the closure of the 
set of all pairs $(R_1,R_2)$ of nonnegative rates 
that can be achieved by time-sharing
Gaussian inputs, namely, rates $(R_1,R_2)$ that satisfy,
\begin{align} 
R_1 &\leq I(X^*_1;Y|X^*_2,V),
    \label{r1star}\\ 
R_2 &\leq I(X^*_2;Y|X^*_1,V),
    \label{r2star}\\ 
R_1 + R_2 &\leq I(X^*_1,X^*_2;Y|V),
    \label{r1r2sumstar}
\end{align}
for some triplet $(V,X^*_1,X^*_2)$ of random variables 
with $\mathbb{E}(X^{*2}_i) \leq P_i$ for $i = 1,2$,
with $X^*_1$ and $X^*_2$ being conditionally independent Gaussians given $V$,
with $V$ taking at most $2$ values, and with $(V,X^*_1,X^*_2)$ being 
independent of the noise $Z$.

Then $\mathcal{C}^*(Z;P_1,P_2)$ contains the region defined as $\mathcal{C}(Z;P_1,P_2),$ but with \eqref{r1}, \eqref{r2} and \eqref{r1r2sum}
replaced with,
\begin{align*} 
R_1 &\leq \sum_{v}\mathbb{P}(V=v)
	\Big(\frac{\snr_{1,v}}{3\snr_{1,v} + 2}\Big)
	I(X_1;Y|X_2,V=v), \\ 
R_2 &\leq \sum_{v}\mathbb{P}(V=v)
	\Big( \frac{\snr_{2,v}}{3\snr_{2,v}+ 2}\Big)
	I(X_2;Y|X_1,V=v), \\ 
R_1 + R_2 &\leq\sum_{v}\mathbb{P}(V=v)
	\Big(\frac{\snr_{v}}{3\snr_{v} + 2}\Big)
	I(X_1,X_2;Y|V=v),
\end{align*}
where $\snr_{i,v}=\Var(X_i|V=v)/N$, $i=1,2$, and 
$\snr_{v}=\Var(X_1+X_2|V=v)/N$.
\end{theorem}

\begin{proof}
It suffices to show that, for any triplet $(V,X_1,X_2)$
with $\mathbb{E}(X^2_i) \leq P_i$ for $i = 1,2$, 
$X_1$ and $X_2$ being conditionally independent given $V$, 
and with $(V,X_1,X_2)$ being independent of $Z$, 
we can find a triplet $(V^*,X_1^*,X_2^*)$
with $\mathbb{E}(X^{*2}_i) \leq P_i$ for $i = 1,2$, 
$X_1^*$ and $X_2^*$ being conditionally independent Gaussians 
given $V^*$, and with $(V^*,X_1^*,X_2^*)$ being 
independent of $Z$, such that,
\begin{align} \label{MACWTS1}
I(X^*_1;Y|X^*_2,V) &\geq \sum_{v}\mathbb{P}(V=v)\frac{\Var{(X_1|V=v)}}{3\Var{(X_1|V=v)} + 2N}I(X_1;Y|X_2,V=v) \\  \label{MACWTS2}
I(X^*_2;Y|X^*_1,V) &\geq \sum_{v}\mathbb{P}(V=v)\frac{\Var{(X_2|V=v)}}{3\Var{(X_2|V=v)} + 2N}I(X_2;Y|X_1,V=v)\\  \label{MACWTS3}
 I(X^*_1,X^*_2;Y|V) &\geq \sum_{v}\mathbb{P}(V=v)\frac{\Var{(X_1+X_2|V=v)}}{3\Var{(X_1+X_2|V=v)} + 2N}I(X_1,X_2;Y|V=v).
\end{align}
But this follows by taking $V^*$ an independent copy of $V$ 
and $X_1^*,X_2^*$ conditionally independent Gaussians given $V^*$, 
having the same conditional variances as $X_1|V$ and $X_2|V$. 
Moreover, we can take $(V,X_1^*,X_2^*)$ independent of $Z$. 
Then, clearly, $\mathbb{E}(X^{*2}_i) = \mathbb{E}(X^{2}_i) \leq P_i.$  
Expanding the mutual informations on the left-hand side of \eqref{MACWTS1} as,
$$
 I(X^*_1;Y|X^*_2,V) = \sum_{v}\mathbb{P}(V=v) I(X^*_1;Y|X^*_2,V=v),
$$
and applying~\eqref{eqsnr}, gives~\eqref{MACWTS1};
similarly, expanding the left-hand sides
of \eqref{MACWTS2} and \eqref{MACWTS3} 
and applying~\eqref{eqsnr} gives
\eqref{MACWTS2}--\eqref{MACWTS3}, and the result follows.  
\end{proof}

The conditioning random variable $V$ in the characterisation of the 
capacity region ${\mathcal C}(Z;P_1,P_2)$ in~(\ref{r1})--(\ref{r1r2sum})
is necessary because of convexity: Convex combinations 
of achievable rate pairs are achievable by 
time-sharing~\cite{cover:book2,elgamal:book}. 
Similarly, rate pairs in the region ${\cal C}^*(Z;P_1,P_2)$
defined by~\eqref{r1star}--\eqref{r1r2sumstar} are achievable 
by time-sharing Gaussian codebooks. 
Clearly the inner bound to the capacity region
given by Theorem~\ref{MACimprovement} is a strict improvement 
of the bound implied by one-receiver, one-transmitter per-user case 
of the half-a-bit MIMO-MAC extension~\cite[Theorem 3]{philosof:07} 
in the low-SNR regime. This indicates that,
even in the low-SNR regime, time-sharing of Gaussian codebooks 
can always achieve a ``fractional'' subset of the capacity region. 

\subsection{MIMO channels} 
\label{mimosection}
Consider the linear, additive-noise MIMO channel,
with $d_{t}$ transmitters and $d_{r}$ receivers, 
input $X \in \mathbb{R}^{d_t}$, sum power constraint,
$$
\mathrm{tr}\mathbb{E}XX^T = \mathbb{E}\|X\|^2 \leq P,
$$
and output,
$$Y = \Hrm X+Z \in \R^{d_r},$$
where $\Hrm \in \R^{d_r\times d_t}$ is the channel matrix and 
$Z \in \R^{d_r}$ is arbitrary additive noise with covariance 
matrix $N$,
independent of the input. 
The capacity of this channel, see e.g.~\cite{elgamal:book,philosof:07},
is given by:
\begin{equation} \label{MIMOconstraint}
C(\Hrm,Z;P) = \sup_{X:\mathbb{E}\|X\|^2\leq P}I(X;\Hrm X+N).
\end{equation}
The additive rate loss (in bits) of using uncorrelated Gaussian 
inputs was shown in \cite[Theorem 1]{philosof:07} to be no more than:
\[ 
\begin{cases}\label{caseszamir}
   \frac{d_r}{2}\log_{2}{\bigl(1+\frac{d_t}{d_r}\bigr)},
	& \text{if } d_r\leq d_t,\\
    \frac{d_t}{2},
	&               \text{if } d_r\geq d_t.
\end{cases}
\]

In this section we provide an alternative bound
on the rate loss
which also remains valid at low SNR.
The main tool we will need is the
following straightforward multidimensional extension 
of inequality~\eqref{eqsnr} established in the proof
of Theorem~\ref{zamirimprovement}.

\begin{corollary}
Let $X,X^*$ and $Z$ be independent
random vectors in $\RL^d$ with
covariance matrices $K_X,K_{X^*}$ and $N$, respectively.
If $X^*$ is Gaussian, then:
\begin{equation} 
\label{dimsnr}
I(X^*;X^*+Z) \geq 
\frac{\det{(K_{X^*})}^{\frac{1}{d}}}{2\det{(K_X+N)}^{\frac{1}{d}} 
+ \det{(K_{X^*})}^{\frac{1}{d}}}\times I(X;X+Z).
\end{equation}
\end{corollary}

Using inequality \eqref{dimsnr} we obtain
a multiplicative upper bound on the rate loss 
(with respect to capacity),
induced by using uncorrelated $N(0,P/d_r)$ inputs.
The resulting multiplicative factor
depends on the number of receivers, 
the channel matrix, and the SNR.
The main point of
Theorem~\ref{MIMOimprovement} is that it
improves on the additive bound of~\cite{philosof:07} 
in the low-SNR regime:

\begin{theorem}
\label{MIMOimprovement}
Let $C(\Hrm,Z;P)$ denote the capacity of the linear MIMO channel with channel 
matrix $H \in \R^{d_r\times d_t}$, sum power constraint $P$, and additive 
noise $Z$ having covariance matrix $N$. Let $X^* \sim \mathcal{N}(0,\frac{P}{d_r}\mathrm{I}_{d_r})$ and denote $\snr_{\text{\rm M}} = \frac{P}{\mathrm{tr}{N}}$. Then:
$$
I(X^*;\Hrm X^*+Z) \geq  
\frac{\snr_{\text{\rm M}}\det{(\Hrm\Hrm^T)^{\frac{1}{d_r}}}}{\snr_{\text{\rm M}}\bigl(2\mathrm{tr}(\Hrm\Hrm^T) + \det{(\Hrm\Hrm^T)^{\frac{1}{d_r}}}\bigr)+2}
\times C(\Hrm,Z;P).
$$
\end{theorem}

\begin{proof}
Let $X$ be any input with covariance $K_X$ satisfying \eqref{MIMOconstraint},
and take $X^*$ to be a Gaussian vector with covariance matrix 
$\frac{P}{d_r}\mathrm{I}_{d_r}$. Then $\Hrm X^*$ is a Gaussian with 
covariance $\frac{P}{d_r}\Hrm\Hrm^T$, and applying~\eqref{dimsnr} we get,
\begin{align*}
I(X^*;\Hrm X^*+Z) 
&= I(\Hrm X^*;\Hrm X^*+Z)\\
&\geq \frac{\frac{P}{d_r}\det{(\Hrm\Hrm^T)^{\frac{1}{d_r}}}}
	{2\det{(\Hrm K_X\Hrm^T+N)}^{\frac{1}{d_r}} 
	+ \frac{P}{d_r}\det{(\Hrm\Hrm^T)^{\frac{1}{d_r}}}}
	\times I(\Hrm X;\Hrm X+Z).
\end{align*}
Using that for square, symmetric,
positive semidefinite matrices $A,B$,  
we have $\det{(A)}^{\frac{1}{d_r}} \leq \frac{1}{d_r}\mathrm{tr}{A}$ and 
$\mathrm{tr}(AB) \leq \mathrm{tr}(A)\mathrm{tr}(B)$, yields,
\begin{align*}
I(X^*;\Hrm X^*+Z)&\geq \frac{\frac{P}{d_r}\det{(\Hrm\Hrm^T)^{\frac{1}{d_r}}}}{\frac{2}{d_r}\mathrm{tr}(\Hrm K_X\Hrm^T+N) + \frac{P}{d_r}\det{(\Hrm\Hrm^T)^{\frac{1}{d_r}}}}\times I(\Hrm X;\Hrm X+Z) \\
&= \frac{{P}\det{(\Hrm\Hrm^T)^{\frac{1}{d_r}}}}{{2}\mathrm{tr}(K_X\Hrm^T\Hrm) + {2}\mathrm{tr}(N) + {P}\det{(\Hrm\Hrm^T)^{\frac{1}{d_r}}}}
\times I(\Hrm X;\Hrm X+Z) \\
&\geq  \frac{P\det{(\Hrm\Hrm^T)^{\frac{1}{d_r}}}}{2\mathrm{tr}(K_X)\mathrm{tr}(\Hrm^T\Hrm)+2\mathrm{tr}(N) + P\det{(\Hrm\Hrm^T)^{\frac{1}{d_r}}}}
\times I(\Hrm X;\Hrm X+Z).
\end{align*}
Finally, since $X$ satisfies the sum power 
constraint \eqref{MIMOconstraint},
$$
I(X^*;\Hrm X^*+Z) \geq \frac{P\det{(\Hrm\Hrm^T)^{\frac{1}{d_r}}}}{P\bigl(2\mathrm{tr}(\Hrm^T\Hrm)+\det{(\Hrm\Hrm^T)^{\frac{1}{d_r}}}\bigr)+2\mathrm{tr}(N)}
\times I(X;\Hrm X+Z),
$$
and the result follows upon taking the supremum 
over all inputs $X$ satisfying the constraint. 
\end{proof}

\section*{Acknowledgements}
We are grateful to the two anonymous reviewers for their useful comments,
and particularly to the reviewer who identified a small gap in our original
proof of Theorem~\ref{levystabilityd}.

\newpage

\bibliographystyle{plain}
\bibliography{ik} 

\end{document}